\documentclass[11pt]{article}
\usepackage{amssymb}
\usepackage{amsthm}
\usepackage{amsmath}
\usepackage[dvips]{epsfig}
\usepackage{graphicx}
\usepackage{color}
\usepackage{url}

\usepackage[utf8]{inputenc}
\usepackage[left=1cm, right=2cm, top=2cm, bottom=2cm]{geometry}
\usepackage[english]{babel}
\usepackage{amsfonts}
\usepackage[numbers]{natbib}
\newtheorem{theorem}{\bf Theorem}
\newtheorem{lemma}{\bf Lemma}

\newcommand{\PT}{{\cal PT}}

\title{\bf Stabilization of the coupled pendula chain \\
under parametric $\mathcal{PT}$-symmetric driving force}

\author{E. Destyl$^{1}$, S.P. Nuiro$^{1}$, D.E. Pelinovsky$^{2}$, and P. Poullet$^{1}$ \\
{\small $^{1}$ LAMIA, Université des Antilles, Campus de Fouillole, F-97157 Pointe-à-Pitre Guadeloupe FWI} \\
{\small $^{2}$ Department of Mathematics, McMaster University, Hamilton, Ontario, L8S 4K1, Canada }}

\begin{document}

\maketitle

\begin{abstract}
We consider a chain of coupled pendula pairs, where each pendulum is connected to the nearest neighbors in
the longitudinal and transverse directions. The common strings in each pair are modulated periodically by an external
force. In the limit of small coupling and near the $1:2$ parametric resonance, we derive a novel system of
coupled $\mathcal{PT}$-symmetric discrete nonlinear Schr\"{o}dinger equation, which has Hamiltonian symmetry but
has no gauge symmetry. By using the conserved energy, we find the parameter range for the linear and nonlinear stability
of the zero equilibrium. Numerical experiments illustrate how destabilization of the zero equilibrium
takes place when the stability constraints are not satisfied. The central pendulum excites nearest pendula and
this process continues until a dynamical equilibrium is reached where each pendulum in the chain oscillates at a finite amplitude.
\end{abstract}

\section{Introduction}

Coupled pendula models are fundamental in the theoretical physics as they can be used for modelling of many interesting
physical phenomena, ranging from DNA dynamics to crystal structure of solid states \cite{Kivshar}. Synchronization of
coupled pendula and destabilization of their dynamics due to various parametric forces have been studied in
many details \cite{Pikovsky}. A recent progress towards analytical studies of such models complemented with robust numerical methods
was achieved by reducing the second-order Newton's equations to the amplitude equations of the
discrete nonlinear Schr\"{o}dinger (dNLS) type \cite{Kevrekidis}.

Parametric resonance in a chain of coupled pendula due to a horizontally shaken pendulum chain
was studied experimentally and analytically in \cite{Xu1,Xu2,Xu3}. Numerical approximations were employed
in these papers in order to characterize dynamics of coupled pendula in the presence of bistability.
Another example of recent studies of parametrically driven chains of coupled pendula can be found
in \cite{Susanto1,Susanto2}, where existence and stability of discrete breathers have been addressed
numerically.

It was realized in \cite{barashenkov,barashenkovInt} that the parametrically driven
coupled pendula can be analytically studied with Hamiltonian systems of the dNLS type in the presence of gains and loses.
Such systems can be formulated mathematically by using the concept of parity ($\mathcal{P}$) and time-reversal
($\mathcal{T}$) symmetries, which was used first to characterize the non-Hermitian Hamiltonians \cite{bender}
and has now been widely observed in many physical experiments \cite{benderExp,schindlerExp}.
The presence of Hamiltonian formulation for a class of $\mathcal{PT}$-symmetric dNLS equations
allows to employ methods of Hamiltonian dynamics to characterize stability and long-time dynamics
of breathers in the parametrically driven chains of coupled oscillators \cite{ChernPel1,ChernPel2}.

The main goal of our work is to derive and to study a novel model of the coupled $\mathcal{PT}$-symmetric
dNLS equations which describes parametrically driven chains of the coupled pendula pairs
connected to the nearest neighbors in the longitudinal and transverse directions. See Fig. \ref{fig-chain}
for a graphical illustration. Compared to the recent work in \cite{ChernPel1}, where a similar model was derived in the presence
of {\em direct couplings} between the two pendula in a pair, we consider {\em diffusive couplings}
between the pendula, in terminology of Chapter 8.2 in \cite{Pikovsky}. This diffusive coupling
is better suited to describe interactions between the two pendula connected to each other by a common horizontal string.

\begin{figure}[h!]
\centering
\includegraphics[scale=0.75]{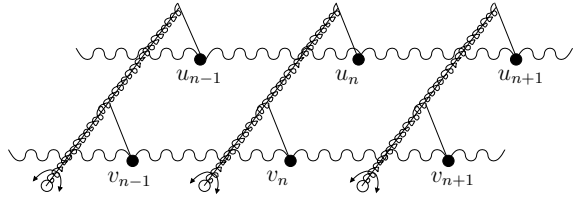}
\caption{A graphical illustration for the chain of coupled pendula connected by torsional springs,
where each pair is hung on a common horizontal string.}
\label{fig-chain}
\end{figure}

The new feature of the coupled $\mathcal{PT}$-symmetric dNLS model derived here is that
the model is Hamiltonian but not gauge-invariant, compared to the previously considered
class of $\mathcal{PT}$-symmetric models \cite{barashenkovInt,ChernPel1}. As a result,
the model admits only one integral of motion, given by the conserved energy of the system. Nevertheless,
existence of this integral of motion allows us to characterize analytically the linear and nonlinear
stability of the zero equilibrium in the system. In the stability region, the chain of coupled pendula
perform stable oscillations which are determined by the initial deviations of the pendula from
the equilibrium.

In order to characterize the destabilization of the coupled pendula chain in the instability region, we employ
numerical experiments. We will show numerically that the central pendulum excites nearest pendula and
this process continues until a dynamical equilibrium is reached where each pendulum in the chain oscillates
at a finite amplitude. While the $\ell^{\infty}$ norm of the oscillation amplitudes remains finite
at the dynamical equilibrium, we show
numerically that the $\ell^2$ norm of the oscillation amplitudes diverges in an unbounded chain of coupled pendula
as the time goes to infinity. The question of whether the oscillation amplitudes can diverge
to infinity in the $\mathcal{PT}$-symmetric dNLS models have been a subject of active research, see, e.g.,
\cite{pel1,pel3,pel4,pascal1}.

The paper is organized as follows. Section 2 describes Newton's equations of motion for the chain of coupled
pendula and the derivation of the $\mathcal{PT}$-symmetric dNLS model. Section 3 addresses linear stability of the zero equilibrium.
Section 4 contains the main results on nonlinear stability of the coupled pendula.
Section 5 reports outcomes of the numerical experiments. Section 6 concludes the paper.

\section{Derivation of the $\mathcal{PT}$-symmetric dNLS model}
\label{model}

We consider a chain of coupled pendula displayed on Fig. \ref{fig-chain}, where each pendulum is
connected to the nearest neighbors in the longitudinal and transverse directions.
In the case of the diffusive couplings between the pendula, Newton's equations of motion are given by
\begin{eqnarray}
\left\{ \begin{array}{l} \ddot{x}_n + \sin(x_n) = C \left( x_{n+1} - 2x_n + x_{n-1} \right) + D (y_n - x_n), \\
 \ddot{y}_n + \sin(y_n) = C \left( y_{n+1} - 2 y_n + y_{n-1} \right) + D (x_n-y_n),
\end{array} \right. \quad n \in \mathbb{Z}, \quad t \in \mathbb{R},
\label{oscillators}
\end{eqnarray}
where $(x_n,y_n)$ correspond to the angles in each pair of the two pendula, dots denote derivatives with respect
to time $t$, and the positive parameters $C$ and $D$ describe couplings between the nearest pendula
in the longitudinal and transverse directions, respectively.

Dynamics of coupled pendula is considered under the following simplifying assumptions:
\begin{itemize}
\item[(A1)] The coupling parameters $C$ and $D$ are small.

\item[(A2)] A uniform periodic force is applied to the common strings for each pair of coupled pendula
and the frequency of the periodic force is picked at the $1:2$ resonance with the linear frequency of each pendulum.
\end{itemize}

According to (A1), we introduce a small parameter $\mu$ such that
both $C$ and $D$ are proportional to $\mu^2$. According to (A2),
we consider $D$ to be proportional to $\cos(2 \omega t)$, where
$|\omega-1|$ is proportional to $\mu^2$. We denote the proportionality
coefficients by $\epsilon$, $\gamma$, and $\Omega$, respectively, hence
parameters $C$ and $D$ are given by
\begin{equation}
C = \epsilon \mu^2, \quad D(t) = 2 \gamma \mu^2 \cos(2 \omega t), \quad \omega^2 = 1 + \mu^2 \Omega,
\end{equation}
where $\gamma, \epsilon, \Omega$ are $\mu$-independent parameters.
In the formal limit $\mu \to 0$, the pendula are uncoupled, and
their small-amplitude oscillations can be studied with the asymptotic multi-scale expansion
\begin{equation}
        \left\{ \begin{array}{l} x_n(t) = \mu \left[ A_n(\mu^2 t) e^{i \omega t} + \bar{A}_n(\mu^2 t) e^{-i \omega t} \right] + \mu^3 X_n(t;\mu), \\[3pt]
y_n(t) = \mu \left[ B_n(\mu^2 t) e^{i \omega t} + \bar{B}_n(\mu^2 t) e^{-i \omega t} \right] + \mu^3 Y_n(t;\mu),
\end{array} \right.
\label{expansion}
\end{equation}
where $(A_n,B_n)$ are amplitudes for nearly harmonic oscillations and $(X_n,Y_n)$ are remainder terms.
Rigorous justification of the asymptotic expansions (\ref{expansion}) in a similar context has been developed in \cite{PPP},
see also \cite{barashenkovInt,ChernPel1}.

From the conditions that the remainder terms $(X_n,Y_n)$ remain bounded as the system evolves, it can be shown
by straightforward computations that the amplitudes $(A_n,B_n)$ satisfy the discrete nonlinear
Schr\"{o}dinger (dNLS) equations in the following form:
\begin{eqnarray}
\left\{ \begin{array}{l} i \frac{d A_n}{d\tau} = \epsilon \left( A_{n+1} - 2A_n + A_{n-1} \right) + \Omega A_n
+ \gamma (\bar{B}_n - \bar{A}_n) + \frac{1}{2} |A_n|^2 A_n, \\
i \frac{d B_n}{d\tau} = \epsilon \left( B_{n+1} - 2 B_n + B_{n-1} \right) + \Omega B_n
+ \gamma (\bar{A}_n - \bar{B}_n) + \frac{1}{2} |B_n|^2 B_n,
\end{array} \right. \quad n \in \mathbb{Z},
\label{dNLS}
\end{eqnarray}
where $\tau = \frac{1}{2} \mu^2 t$.
The system (\ref{dNLS}) takes the form of coupled parametrically driven dNLS equations.
There exists an invariant reduction of system (\ref{dNLS}) to the scalar dNLS equation \cite{Kevrekidis}
if
\begin{equation}
\label{synchr-1}
A_n = B_n, \quad n \in \mathbb{Z}.
\end{equation}
The reduction (\ref{synchr-1}) corresponds to the synchronization in each pair of coupled pendula with
\begin{equation}
\label{synchr-2}
x_n = y_n, \quad n \in \mathbb{Z},
\end{equation}
when the periodic driving force does not affect the dynamics of the coupled pendula chain.

Unless the reduction (\ref{synchr-1}) is imposed, the system of coupled dNLS equations (\ref{dNLS}) is not
invariant with respect to a gauge transformation of the amplitudes $(A,B)$. It is however invariant with respect
to the exchange $A \leftrightarrow B$. The system can be written in a complex Hamiltonian form:
\begin{equation}
i \frac{d A_n}{d \tau} = \frac{\partial H}{\partial \bar{A}_n}, \quad
i \frac{d B_n}{d \tau} = \frac{\partial H}{\partial \bar{B}_n}, \quad n \in \mathbb{Z},
\label{Hamiltonian-dNLS}
\end{equation}
associated with the conserved energy function
\begin{eqnarray}
\nonumber
H(A,B) & = & \sum_{n \in \mathbb{Z}} \frac{1}{4} (|A_n|^4 + |B_n|^4) + \Omega (|A_n|^2 + |B_n|^2) + \gamma ( A_n B_n + \bar{A}_n \bar{B}_n) \\
 & \phantom{t} & - \frac{1}{2} \gamma (A_n^2 + \overline{A}_n^2 + B_n^2 + \overline{B}_n^2) - \epsilon |A_{n+1}-A_n|^2 - \epsilon |B_{n+1}-B_n|^2.
\label{energy-dNLS}
\end{eqnarray}
The Hamiltonian structure (\ref{Hamiltonian-dNLS}) and the conserved energy function (\ref{energy-dNLS}) are inherited
from the Hamiltonian structure of the Newton's equations of motion (\ref{oscillators}) after the substitution
of the asymptotic expansion (\ref{expansion}) and its truncation.

The system (\ref{dNLS}) can be cast to the form of the parity--time reversal ($\PT$) dNLS equations
\cite{barashenkovInt,ChernPel1}. Using the variables
\begin{equation}
\label{variables}
u_n := \frac{1}{4} \left( A_n - i \bar{B}_n \right), \quad v_n := \frac{1}{4} \left( A_n + i \bar{B}_n \right),
\end{equation}
the system of coupled dNLS equations (\ref{dNLS}) can be cast to the equivalent form
\begin{eqnarray}
\left\{ \begin{array}{l} i \frac{d u_n}{d \tau} = \epsilon \left( v_{n+1} - 2 v_n + v_{n-1} \right)
+ \Omega v_n + i \gamma u_n - \gamma \bar{u}_n +
                2\left[ \left( 2|u_n|^2 + |v_n|^2 \right) v_n + u_n^2 \bar{v}_n \right], \\[3pt]
i \frac{d v_n}{d \tau} = \epsilon \left( u_{n+1} - 2 u_n + u_{n-1} \right)
+ \Omega u_n - i \gamma v_n - \gamma \bar{v}_n +
2\left[ \left( |u_n|^2 + 2 |v_n|^2 \right) u_n + \bar{u}_n v_n^2 \right], \end{array} \right.
\label{PT-dNLS}
\end{eqnarray}
which is invariant with respect to the action of the parity $\mathcal{P}$ and time-reversal
$\mathcal{T}$ operators given by
\begin{equation}
\label{PT-symmetry}
\mathcal{P} \left[ \begin{array}{c} u \\ v \end{array} \right] =
\left[ \begin{array}{c} v \\ u \end{array} \right], \qquad
\mathcal{T} \left[ \begin{array}{c} u(t) \\ v(t) \end{array} \right] =
\left[ \begin{array}{c} \bar{u}(-t) \\ \bar{v}(-t) \end{array} \right].
\end{equation}
Therefore, the system (\ref{PT-dNLS}) can be referred to as the
$\mathcal{PT}$-symmetric dNLS equation. The system (\ref{PT-dNLS}) is still Hamiltonian
with a modified energy functional but it is not gauge-invariant. Unlike the works \cite{ChernPel1,ChernPel2},
where the formulation of a similar system as the $\mathcal{PT}$-symmetric dNLS equation has been explored,
we will work with the dNLS system (\ref{dNLS}) without transforming it to the $\mathcal{PT}$-symmetric
dNLS equation (\ref{PT-dNLS}).

\section{Linear stability of zero equilibrium}

Parametrically driving forces can destabilize the zero equilibrium state in the coupled pendula chain.
We shall first clarify conditions for the linear stability of the zero equilibrium.
The following lemma provides a sharp bound on the parameter $\gamma$ of the driving force which ensures that the zero
equilibrium is linearly stable.

\begin{lemma}
The zero equilibrium of the system (\ref{dNLS}) with $\epsilon > 0$
is linearly stable if $|\gamma| < \gamma_0$, where
\begin{equation}
\label{gamma-0}
\gamma_0 := \left\{ \begin{array}{ll} \frac{1}{2} (\Omega - 4 \epsilon), & \Omega > 4 \epsilon, \\
\frac{1}{2} |\Omega|, & \Omega < 0. \end{array} \right.
\end{equation}
The zero equilibrium is linearly unstable if $|\gamma| \geq \gamma_0$ or if $\Omega \in [0,4 \epsilon]$ and $\gamma \neq 0$.
\label{lemma-equilibrium}
\end{lemma}

\begin{proof}
Truncating the system (\ref{dNLS}) at the linear terms
and using the Fourier transform
\begin{equation}
\label{discrete-Fourier}
A_n = \frac{1}{2\pi} \int_{-\pi}^{\pi} \hat{A} e^{i n \theta} d \theta, \quad n \in \mathbb{Z},
\end{equation}
we obtain the following system of differential equations for $\hat{A}$ and $\hat{B}$
parameterized by $\theta$:
\begin{eqnarray}
\left\{ \begin{array}{l} i \frac{d \hat{A}}{d\tau} + \kappa \hat{A} =
\gamma (\hat{\bar{B}} - \hat{\bar{A}}), \\
i \frac{d \hat{B}}{d\tau} + \kappa \hat{B} =
\gamma (\hat{\bar{A}} - \hat{\bar{B}}),
\end{array} \right.
\label{dNLSlinear}
\end{eqnarray}
where $\kappa := 4 \epsilon \sin^2(\theta/2) - \Omega$. Separating the time
variable like in $\hat{A}(\tau) = \hat{a} e^{i \omega \tau}$ yields the linear homogeneous system
$$
\left[ \begin{matrix}  \kappa - \omega & \gamma & 0 & - \gamma \\
\gamma & \kappa + \omega & -\gamma & 0 \\
0 & -\gamma & \kappa - \omega & \gamma \\
-\gamma & 0 & \gamma &  \kappa + \omega \end{matrix} \right]
  \left[ \begin{array}{c} \hat{a} \\ \hat{\overline{a}} \\ \hat{b}\\ \hat{\overline{b}}
  \end{array} \right] = \left[ \begin{array}{c} 0 \\ 0 \\ 0 \\ 0 \end{array} \right],
$$
with the characteristic equation factorized in the form
\begin{equation}
\label{dispersion-relation}
(\omega^2 - \kappa^2) (\omega^2 - \kappa^2 + 4 \gamma^2) = 0.
\end{equation}
The first pair of roots is always real:
$$
\omega_1^{\pm} = \pm |\kappa| = \pm \left| 4 \epsilon \sin^2(\theta/2) - \Omega \right|.
$$
The second pair of roots is real if $|\kappa| > 2 |\gamma|$:
$$
\omega_2^{\pm} = \pm \sqrt{\kappa^2 - 4 \gamma^2} = \pm \sqrt{(4 \epsilon \sin^2(\theta/2) - \Omega)^2 - 4 \gamma^2}.
$$
This happens for $2 |\gamma| < \Omega - 4 \epsilon$ if $\Omega > 4 \epsilon$
and for $2 |\gamma| < |\Omega|$ if $\Omega < 0$. In both cases, the zero equilibrium is linearly stable.

On the other hand, for any $|\gamma| > \gamma_0$, the values of $\omega_2^{\pm}$ are purely imaginary either near
$\theta = \pm \pi$ if $\Omega > 4 \epsilon$ or near $\theta = 0$ if $\Omega < 0$. In these cases,
the zero equilibrium is linearly unstable with exponentially growing perturbations.
For $|\gamma| = \gamma_0$, the values of $\omega_2^{\pm}$ are zero either at $\theta = \pm \pi$ if $\Omega > 4 \epsilon$ or
at $\theta = 0$ if $\Omega < 0$. The zero equilibrium is linearly unstable with polynomially growing perturbations.

Finally, if $\Omega \in [0, 4\epsilon]$, there exists $\theta_0 \in [-\pi,\pi]$
such that $\kappa = 0$. Then, for any $\gamma \neq 0$, the value of $\omega_2^{\pm}$ are purely imaginary
and the zero equilibrium is linearly unstable with exponentially growing perturbations.
\end{proof}

\section{Nonlinear stability of coupled pendula}

The Cauchy problem for the system of coupled dNLS equations (\ref{dNLS}) can be posed in
sequence space $\ell^2(\mathbb{Z})$. Local existence of solutions to the Cauchy problem
in $\ell^2(\mathbb{Z})$ follows from an easy application of Picard's method. Combining with
the a priori energy estimates yields global existence of solutions.
Global existence follows by the energy method. The following lemma states the global
well-posedness result.

\begin{lemma}
\label{proposition-existence}
For every $(A^{(0)},B^{(0)}) \in \ell^2(\mathbb{Z})$, there exists a unique solution $(A,B)(\tau) \in
C^1(\mathbb{R},\ell^2(\mathbb{Z}))$ of the system of coupled dNLS equations (\ref{dNLS}) such that
$(A,B)(0) = (A^{(0)},B^{(0)})$. The unique solution depends continuously on initial data
$(A^{(0)},B^{(0)}) \in \ell^2(\mathbb{Z})$.
\end{lemma}

\begin{proof}
The Cauchy problem for the system of coupled dNLS equations (\ref{dNLS}) can be posed
by using the system of integral equations
\begin{eqnarray}
\label{integral-eq}
\left\{ \begin{array}{l}
A(\tau) = A^{(0)} -i \int_0^{\tau} F_A(A(\tau'),B(\tau')) d\tau', \\
B(\tau) = B^{(0)} -i \int_0^{\tau} F_B(A(\tau'),B(\tau')) d\tau', \end{array} \right.
\end{eqnarray}
where $(F_A,F_B)$ stand for the right-hand sides of the system (\ref{dNLS}).
Since the discrete Laplacian is a bounded operator in $\ell^2(\mathbb{Z})$ with the bound
\begin{equation}
\label{bounded}
\sum_{n \in \mathbb{Z}} |A_{n+1}-A_n|^2 \leq 4 \| A\|_{\ell^2}^2, \quad A \in \ell^2(\mathbb{Z}),
\end{equation}
and the sequence space $\ell^2(\mathbb{Z})$ forms a Banach algebra with respect to pointwise multiplication with the bound
\begin{equation}
\label{Banach}
\| A B \|_{\ell^2} \leq \| A \|_{\ell^2} \| B \|_{\ell^2}, \quad A, B \in \ell^2(\mathbb{Z}),
\end{equation}
there is a sufficiently small $\tau_0 > 0$ such that the integral equations (\ref{integral-eq})
admit a unique solution $(A,B)(\tau) \in C^0([-\tau_0,\tau_0],\ell^2(\mathbb{Z}))$
with $(A,B)(0) = (A^{(0)},B^{(0)})$ by the contraction mapping method.
By the same method, the solution depends continuously on initial data
$(A^{(0)},B^{(0)}) \in \ell^2(\mathbb{Z})$. Thanks again to the boundedness of the
discrete Laplacian operator in $\ell^2(\mathbb{Z})$,
bootstrap arguments extend this solution in $C^1([-\tau_0,\tau_0],\ell^2(\mathbb{Z}))$.

The local solution is continued globally by using the energy method.
For any solution $(A,B)(\tau)$ in $C^1([-\tau_0,\tau_0],\ell^2(\mathbb{Z}))$, we obtain
the following balance equation from system (\ref{dNLS}):
$$
\frac{d}{d\tau} \sum_{n \in \mathbb{Z}} (|A_n|^2 + |B_n|^2) = i \gamma \sum_{n \in \mathbb{Z}}
\left( 2 A_n B_n - 2 \bar{A}_n \bar{B}_n + \bar{A}_n^2 + \bar{B}_n^2 - A_n^2 - B_n^2 \right).
$$
Integrating this equation in time and applying Gronwall's inequality, we get
\begin{equation}
\label{growth-A-B}
\| A(\tau) \|_{l^2}^2 + \| B(\tau) \|_{l^2}^2 \le \left(\| A^{(0)}\|_{l^2}^2 + \|B^{(0)} \|_{l^2}^2\right)
e^{4 |\gamma \tau|}, \quad \tau \in [-\tau_0,\tau_0].
\end{equation}
Therefore $\| A(\tau)\|_{l^2}$ and $\| B(\tau)\|_{l^2}$ cannot blow up in a finite time
and so cannot their derivatives in $\tau$. Therefore,
the local solution $(A,B)(\tau) \in C^1([-\tau_0,\tau_0],\ell^2(\mathbb{Z}))$
is continued for every $\tau_0$.
\end{proof}

The bound (\ref{growth-A-B}) does not exclude a possible exponential growth of
the $\ell^2(\mathbb{Z})$ norms of the global solution $(A,B)(\tau)$ as $\tau \to \infty$.
However, thanks to coercivity of the energy function (\ref{energy-dNLS}) near the zero
equilibrium, we can still obtain a time-independent bound on the $\ell^2(\mathbb{Z})$ norm of the solution
near the zero equilibrium, provided it is linearly stable. Moreover, for $\Omega > (2|\gamma| + 4 \epsilon)$,
the global bound holds for arbitrary initial data. The following three theorems represent the corresponding
results. For simplicity, we can restrict our attention to $\gamma > 0$.

\begin{theorem}
For every $\Omega > 2\gamma + 4 \epsilon$ and every initial data $(A^{(0)},B^{(0)}) \in \ell^2(\mathbb{Z})$,
there is a positive constant $C$ such that the unique solution $(A,B)(\tau) \in
C^1(\mathbb{R},\ell^2(\mathbb{Z}))$ of the system of coupled dNLS equations (\ref{dNLS}) satisfies
\begin{equation}
\label{time-ind-bound}
\| A(\tau) \|_{\ell^2}^2 + \| B(\tau) \|_{\ell^2}^2 \leq C, \quad \mbox{\rm for every} \;\; \tau \in \mathbb{R}.
\end{equation}
\label{theorem-bound}
\end{theorem}

\begin{proof}
If $\Omega > 2\gamma + 4 \epsilon$, the following lower bound for the energy function $H$ is positive:
\begin{equation}
\label{coercivity-E}
H \ge (\Omega - 2 \gamma - 4 \epsilon) \left(\| A(\tau) \|_{\ell^2}^2 + \| B(\tau) \|_{\ell^2}^2\right), \quad \tau \in \mathbb{R},
\end{equation}
where we have used the Cauchy--Schwarz inequality and the bound (\ref{bounded}) and we have dropped
the positive quartic terms from the lower bound. Since $H$ is constant in $\tau$ and is bounded
for any $(A,B)(\tau) \in C^1(\mathbb{R},\ell^2(\mathbb{Z}))$ due to the continuous embedding
\begin{equation}
\label{embedding}
\| A \|_{\ell^4} \leq \| A \|_{\ell^2}, \quad A \in \ell^2(\mathbb{Z}),
\end{equation}
the time-independent bound (\ref{time-ind-bound}) follows from the lower bound (\ref{coercivity-E}) for any $\Omega > 2\gamma + 4 \epsilon$.
\end{proof}

\begin{theorem}
For every $\Omega < -2\gamma$, there exist $\delta_0 > 0$ and $C_0 > 0$ such that
for every initial data $(A^{(0)},B^{(0)}) \in \ell^2(\mathbb{Z})$ satisfying
\begin{equation}
\label{bound-initial}
\delta := \| A^{(0)} \|_{\ell^2}^2 + \| B^{(0)} \|_{\ell^2}^2 \leq \delta_0,
\end{equation}
the unique solution $(A,B)(\tau) \in
C^1(\mathbb{R},\ell^2(\mathbb{Z}))$ of the system of coupled dNLS equations (\ref{dNLS}) satisfies
\begin{equation}
\label{final-bound}
\| A(\tau) \|_{\ell^2}^2 + \| B(\tau) \|_{\ell^2}^2 \leq C_0 \delta, \quad \mbox{\rm for every} \;\; \tau \in \mathbb{R}.
\end{equation}
\label{theorem-bound-negative}
\end{theorem}

\begin{proof}
If $\Omega < -2 \gamma$, the following lower bound for the energy function $-H$ holds:
\begin{equation}
\label{coercivity-E-pos}
-H \ge (|\Omega| - 2 \gamma) \left(\| A(\tau) \|_{\ell^2}^2 + \| B(\tau) \|_{\ell^2}^2\right)
- \frac{1}{4} \left(\| A(\tau) \|_{\ell^2}^2 + \| B(\tau) \|_{\ell^2}^2\right)^2, \quad \tau \in \mathbb{R},
\end{equation}
where we have used the Cauchy--Schwarz inequality and the bound (\ref{embedding})
and we have dropped the positive discrete Laplacian terms from the lower bound. Since $-H > 0$ is constant in $\tau$
and the initial data in (\ref{bound-initial}) is small, there exists a constant $C > 0$
independently of $\delta$, such that $|H| \leq C \delta$. Then, the lower bound (\ref{coercivity-E-pos})
yields the upper bound (\ref{final-bound}) with the choice of constant $C_0 > 0$ to satisfy
$$
\frac{C \delta}{|\Omega| - 2 \gamma - \frac{1}{4} C_0\delta} \leq C_0 \delta,
$$
which is always possible since $C$ and $C_0$ are $\delta$-independent and $\delta \leq \delta_0$
with sufficiently small $\delta_0$.
\end{proof}

\begin{theorem}
\label{theorem-finite}
If the lattice is truncated on finitely many ($N$) sites, denoted by $\mathbb{Z}_N$, then
for every $\Omega, \gamma, \epsilon$ and every initial data $(A^{(0)},B^{(0)}) \in \ell^2(\mathbb{Z}_N)$,
there is a positive constant $C_N$ such that
the unique solution $(A,B)(\tau) \in
C^1(\mathbb{R},\ell^2(\mathbb{Z}_N))$ of the system of coupled dNLS equations (\ref{dNLS})
on the truncated lattice $\mathbb{Z}_N$ satisfies
\begin{equation}
\label{time-ind-bound-truncated}
\| A(\tau) \|_{\ell^2}^2 + \| B(\tau) \|_{\ell^2}^2 \leq C_N, \quad \mbox{\rm for every} \;\; \tau \in \mathbb{R},
\end{equation}
with $C_N \to \infty$ as $N \to \infty$.
\end{theorem}

\begin{proof}
If the lattice is truncated on $N$ sites, one can use the bound
\begin{equation}
\label{embedding-reversed}
\| A \|_{\ell^2} \leq N^{1/4} \| A \|_{\ell^4}, \quad A \in \ell^4(\mathbb{Z}_N),
\end{equation}
and obtain a different lower bound for the energy function $H$:
$$
H \ge \frac{1}{4} \left(\| A(\tau) \|_{\ell^4}^4 + \| B(\tau) \|_{\ell^4}^4\right) - (|\Omega| + \gamma + 4 \epsilon) N^{1/2}
\left(\| A(\tau) \|_{\ell^4}^2 + \| B(\tau) \|_{\ell^4}^2\right), \quad \tau \in \mathbb{R}.
$$
Since $H$ is constant in $\tau$, the time-independent
bound (\ref{time-ind-bound-truncated}) follows from the lower bound above with $C_N < \infty$.
However, $C_N \to \infty$ as $N \to \infty$.
\end{proof}

\section{Numerical experiments}

We present here numerical approximations of dynamics of the dNLS system (\ref{dNLS}). For simplicity,
all simulations correspond to the choice $\gamma = 1$ and $\epsilon = 1$. The lattice is truncated
on $N$ sites with an odd choice for $N$. The resulting finite system of differential equations is approximated by
the MATLAB solver {\em ode45} with relative tolerance set at $10^{-7}$ and absolute tolerance set at $10^{-9}$.

\begin{figure}[htb]
\begin{center}
\includegraphics[height=5cm]{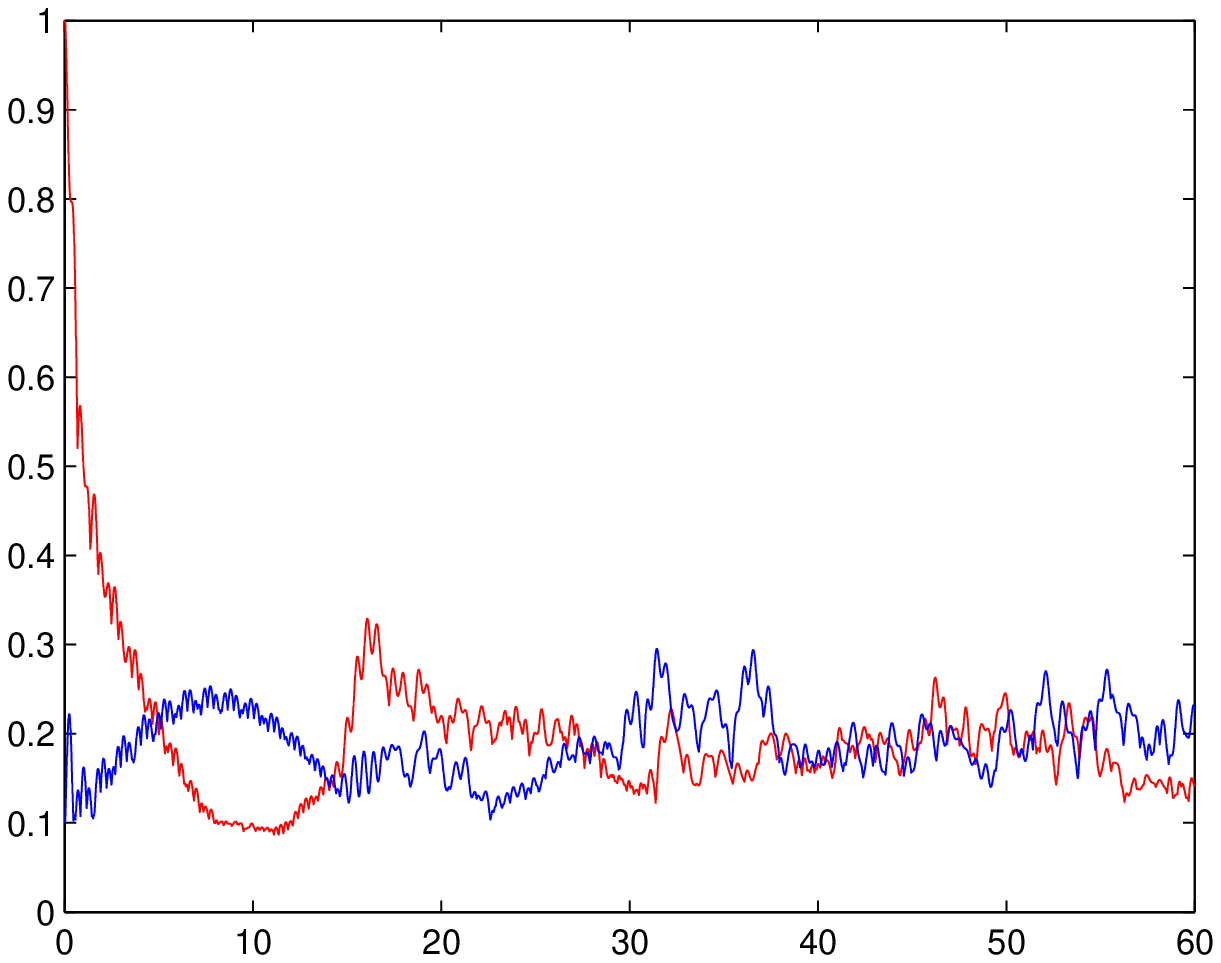}\includegraphics[height=5cm]{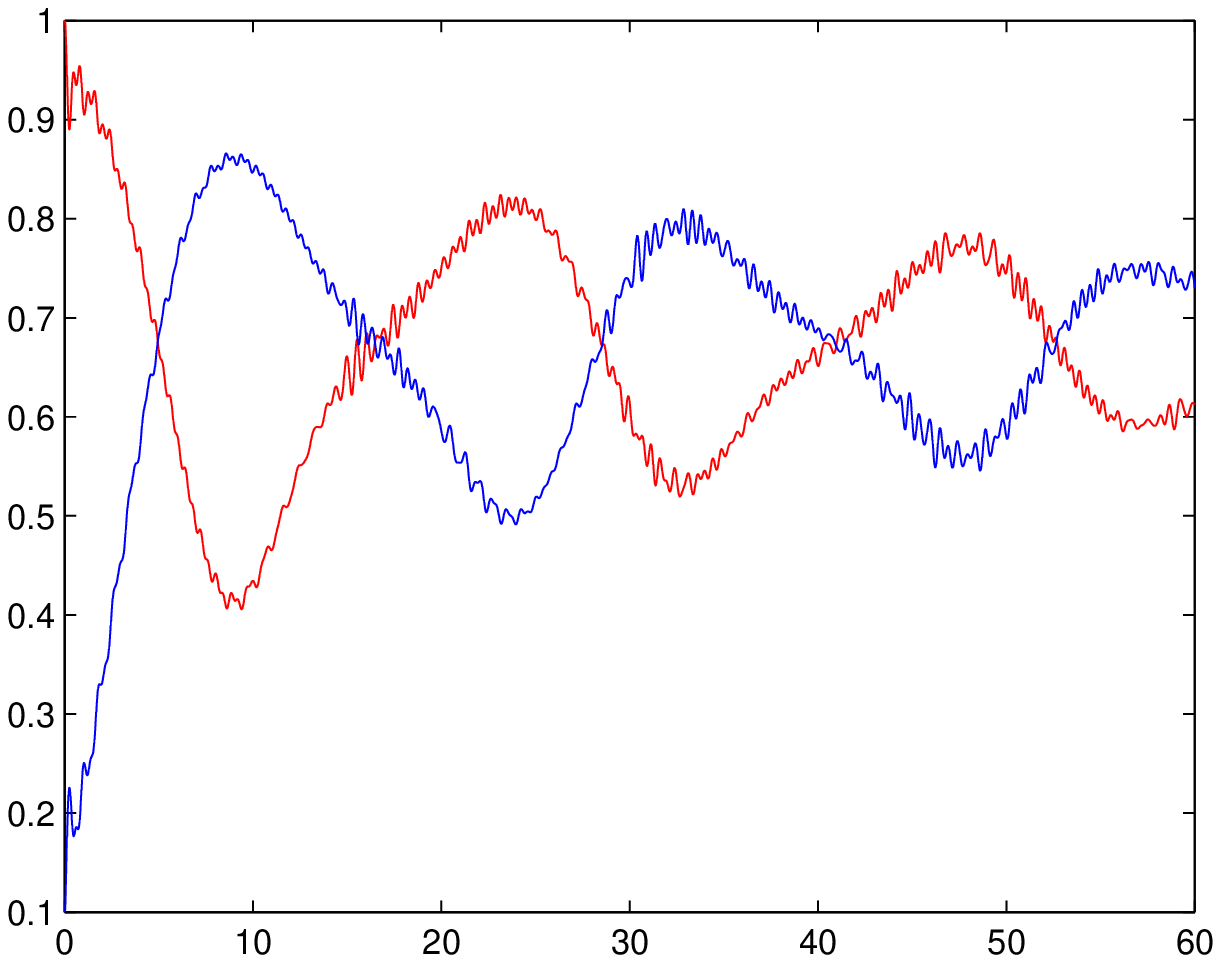} \\
\includegraphics[height=5cm]{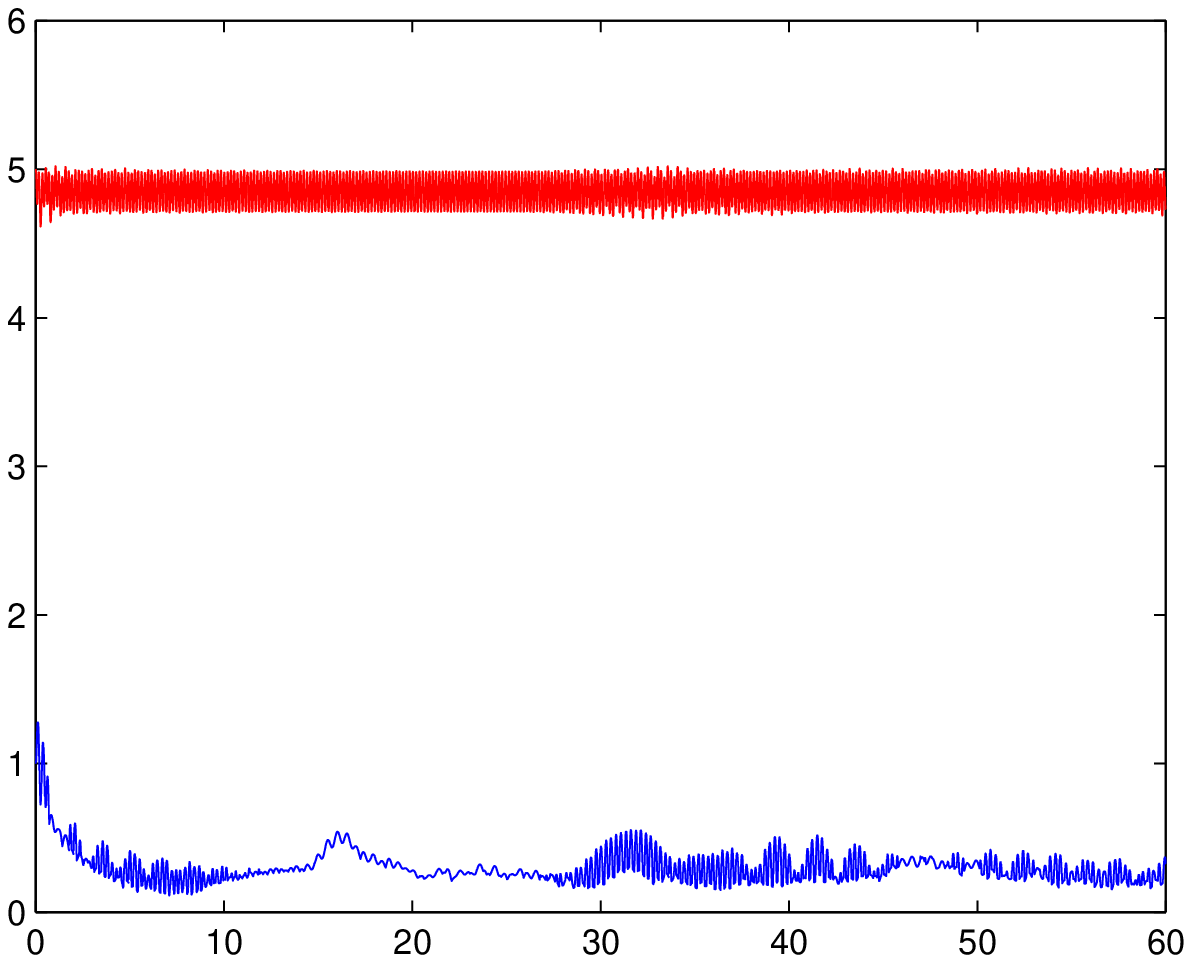}\includegraphics[height=5cm]{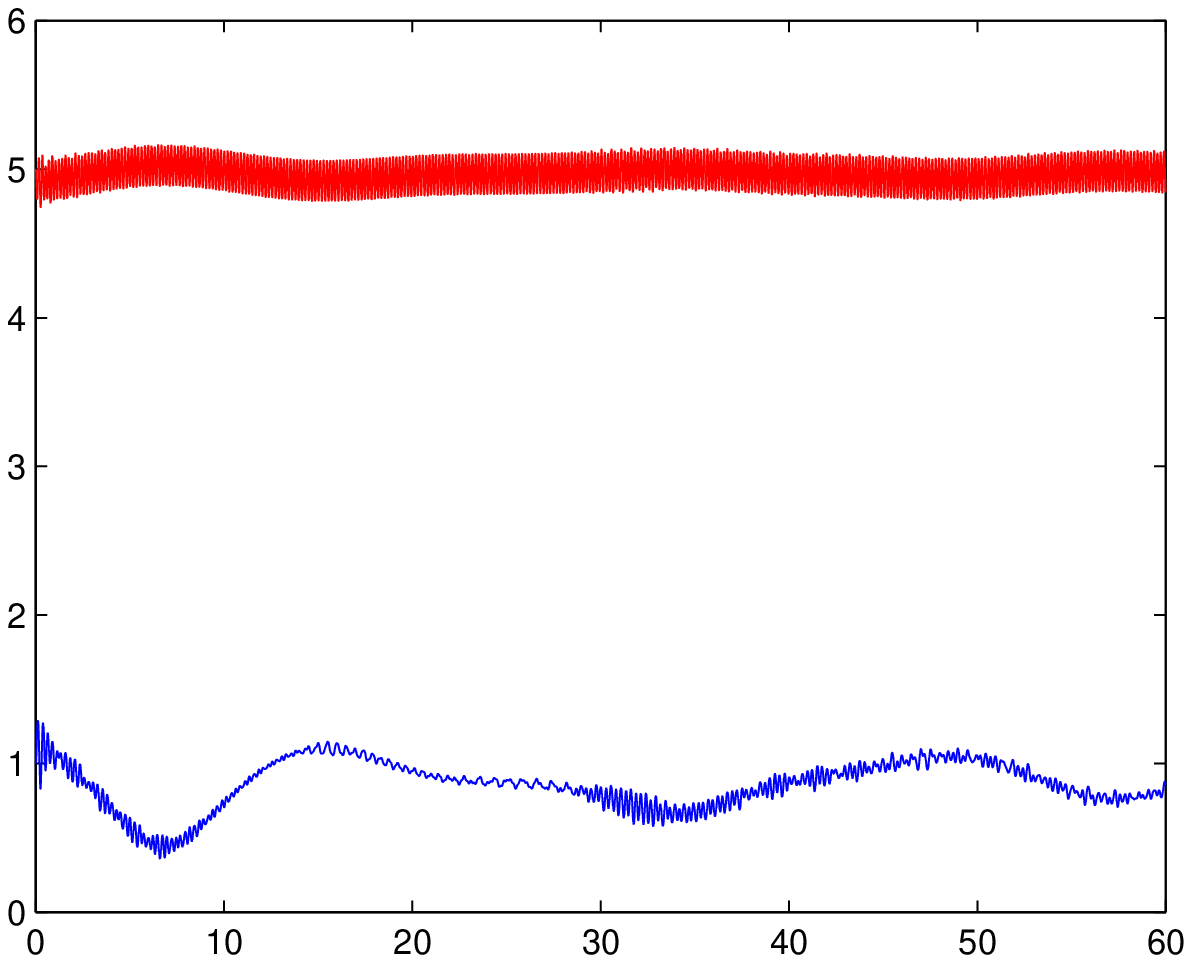}
\end{center}
\caption{$\ell^{\infty}$ (left) and $\ell^2$ (right) norms of the solution for $A$ (red) and $B$ (blue) for $\Omega = 8$:
upper panels for $(A_0,B_0) = (1,0.1)$ and lower panels for $(A_0,B_0) = (5,1)$.}
\label{fig-1}
\end{figure}

Figures \ref{fig-1}--\ref{fig-3} show the dependence of the $\ell^{\infty}$ (left) and $\ell^2$ (right) norms
of the components $A$ (red) and $B$ (blue) in the chain of $N = 61$ oscillators
for three characteristic values of the parameter $\Omega$. The upper panels show computations
for the initial data which are nonzero on the central site with $(A_0,B_0) = (1,0.1)$.
The lower panels show computations with $(A_0,B_0) = (5,1)$.
As a measure of the numerical error, we control conservation of energy $H$ in (\ref{energy-dNLS}).
We have observed that the energy is conserved up to the order of $10^{-6}$ for smaller initial data
(upper panels) and up to the order of $10^{-4}$ for larger initial data (lower panels).

For $\Omega = 8 > 2 \gamma + 4 \epsilon$,
the result of Theorem \ref{theorem-bound} provides a global bound on the $\ell^2(\mathbb{Z})$ norm of the solution.
This is confirmed by the numerical simulation on Figure \ref{fig-1}.
The upper panels show that for smaller initial conditions,
the central pendulum excites other pendula in the chain, this process
results in the decrease of the oscillation amplitude of the central pendulum.
While the $\ell^2$ norm oscillates between the modes $A$ and $B$, the $\ell^{\infty}$ norms of both modes $A$ and $B$
are comparable and do not change much in the time evolution after an initial time interval.
The lower panels show that for larger initial conditions, dynamics of the $A$ mode is effectively separated
from dynamics of the $B$ modes and the central pendulum does not excite large oscillations of other pendula.

\begin{figure}[htb]
\begin{center}
\includegraphics[height=5cm]{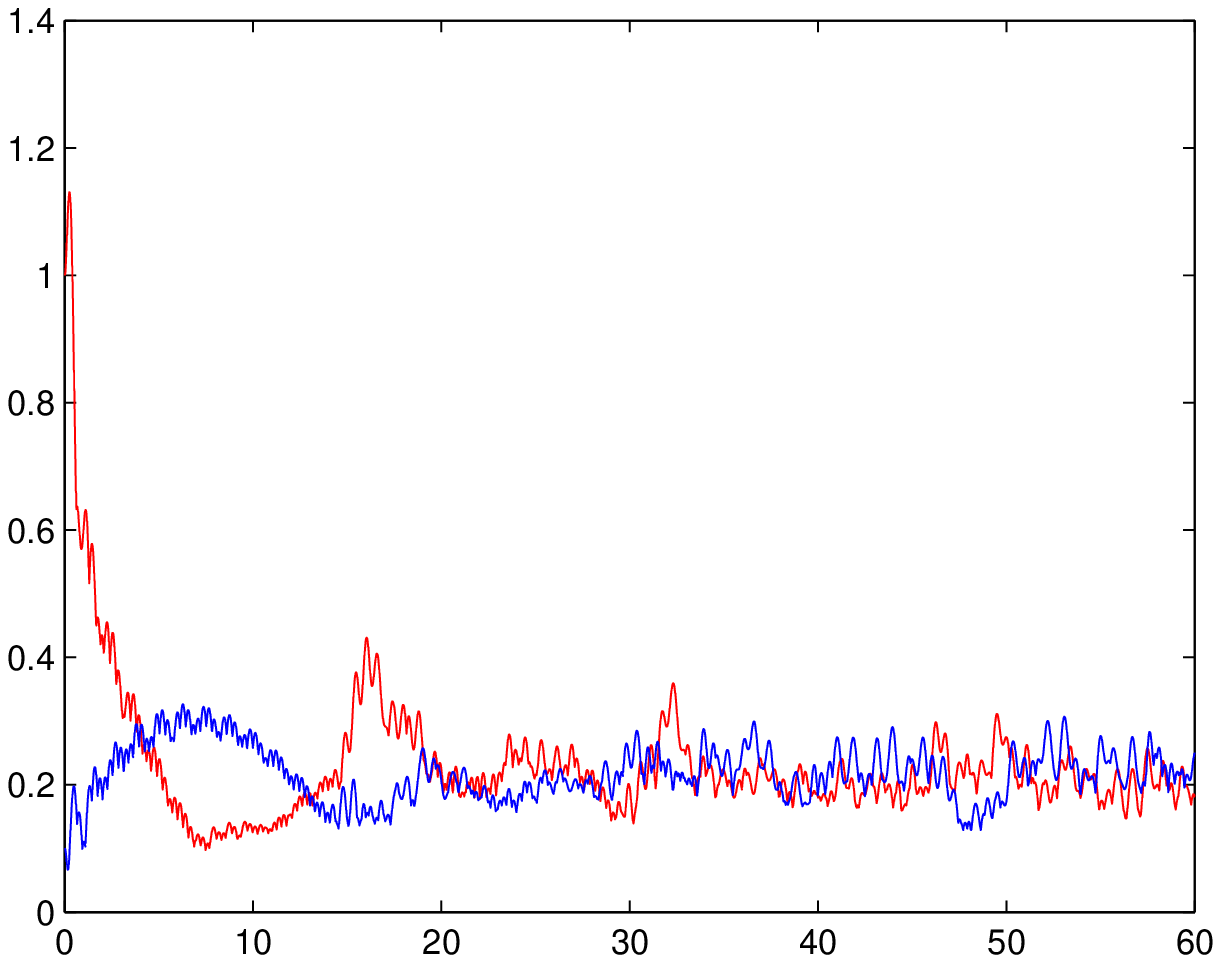} \includegraphics[height=5cm]{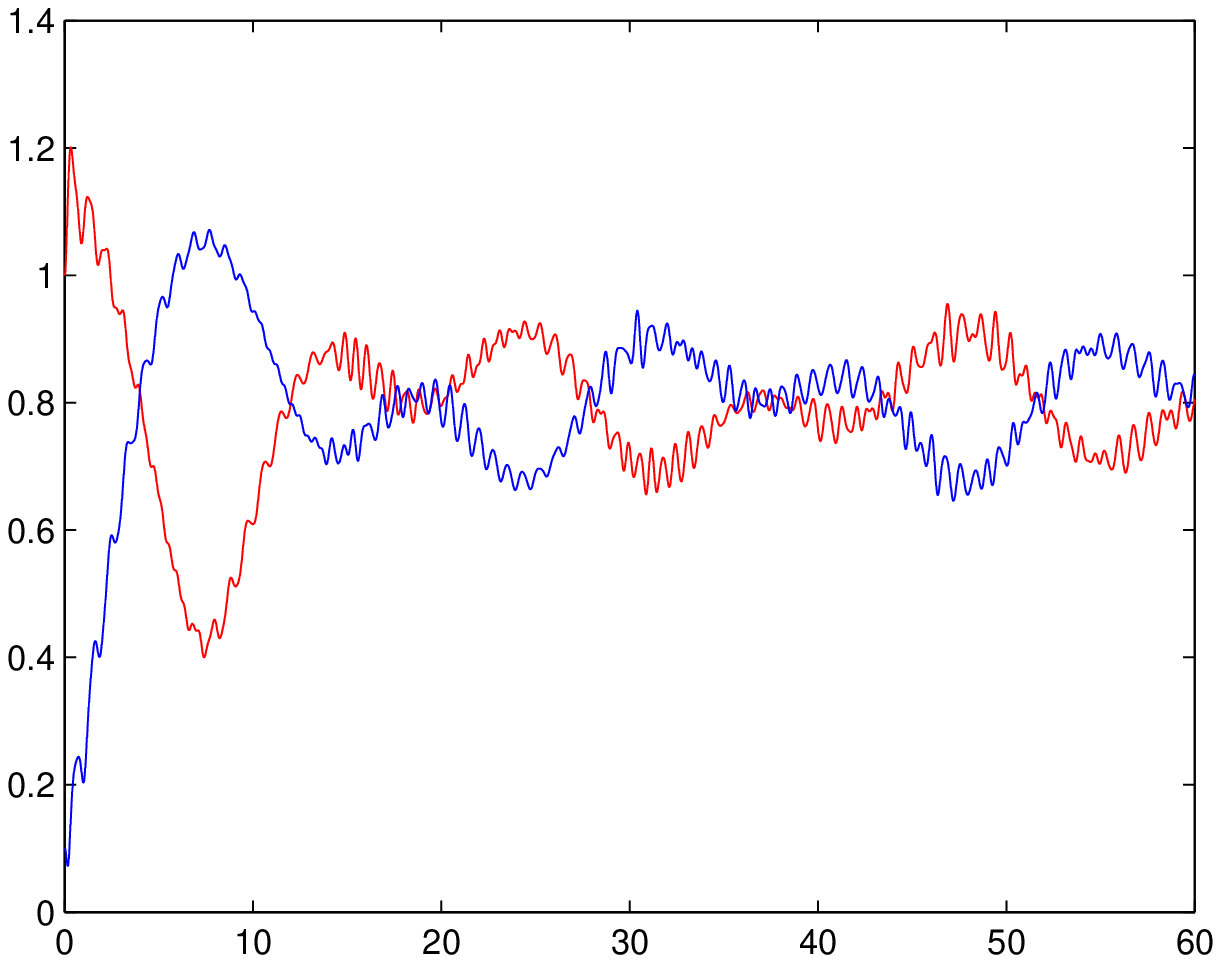}\\
\includegraphics[height=5cm]{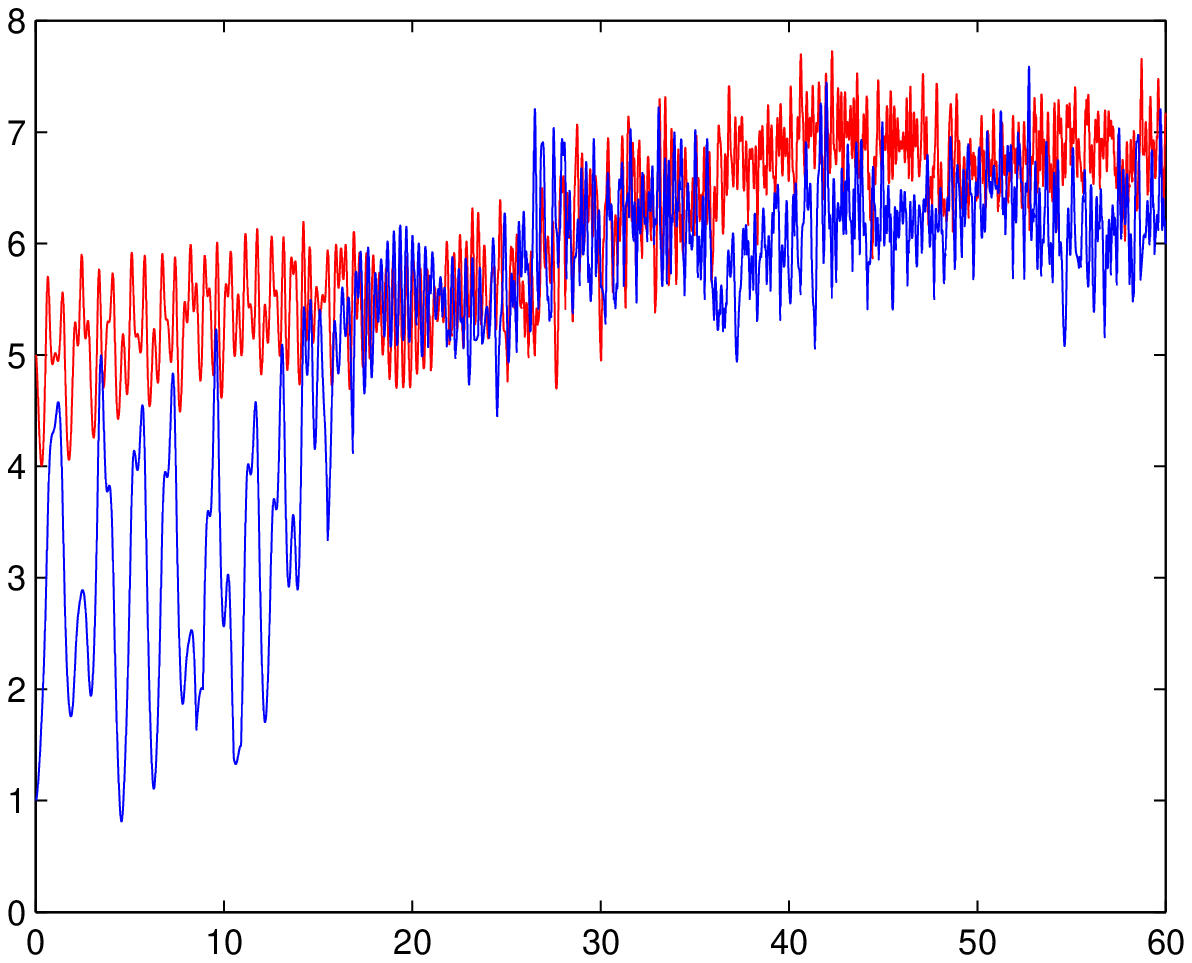} \includegraphics[height=5cm]{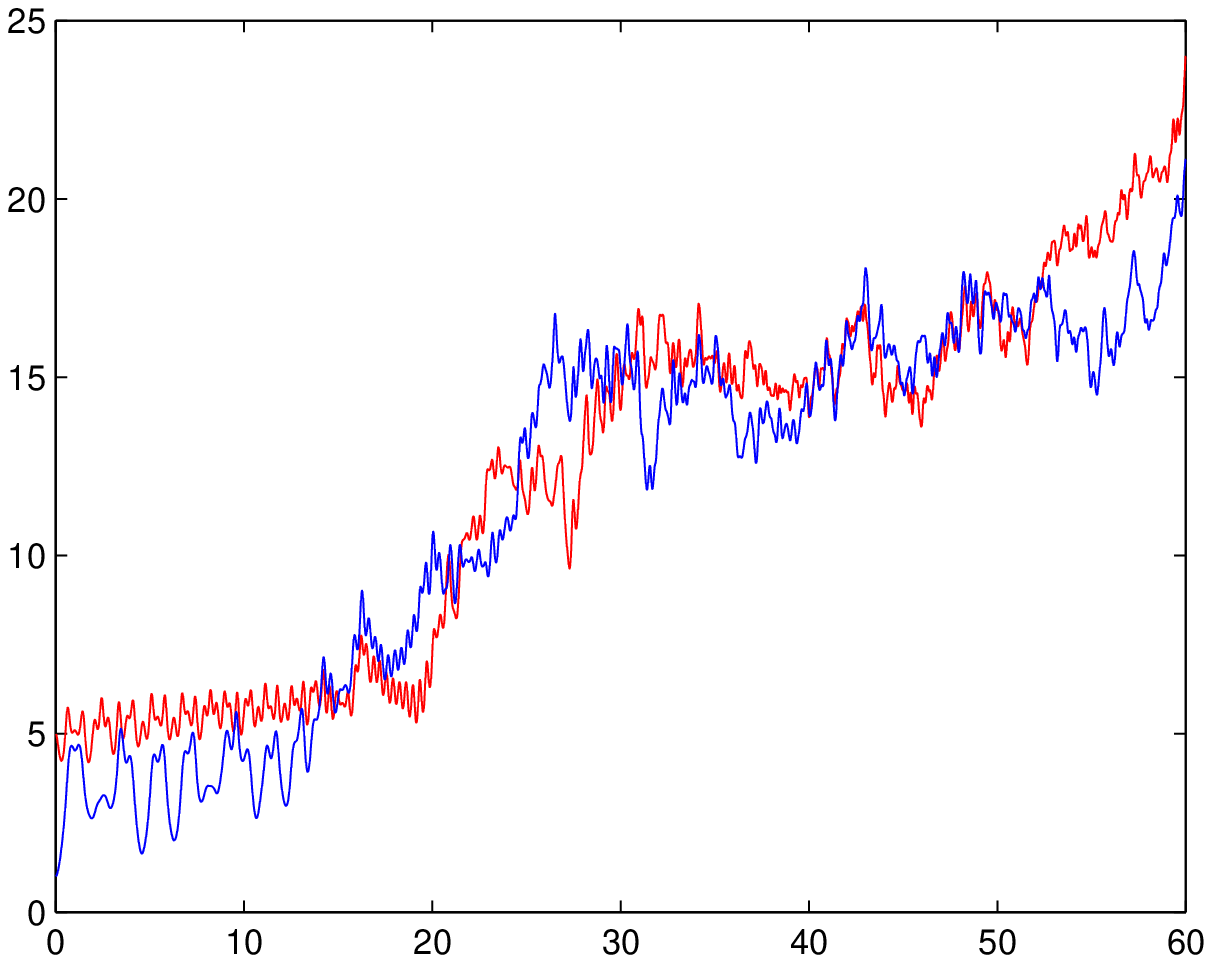}
\end{center}
\caption{$\ell^{\infty}$ (left) and $\ell^2$ (right) norms of the solution for $A$ (red) and $B$ (blue) for $\Omega = -4$:
upper panels for $(A_0,B_0) = (1,0.1)$ and lower panels for $(A_0,B_0) = (5,1)$.}
\label{fig-2}
\end{figure}

For $\Omega = -4 < -2 \gamma$,
the result of Theorem \ref{theorem-bound-negative} provides a global bound on the $\ell^2(\mathbb{Z})$ norm of the solution
only in the case of small initial conditions. For smaller initial conditions,
the numerical simulations on the upper panels of Figure \ref{fig-2} do not show drastic differences
compared to the case $\Omega = 8 > 2 \gamma + 4 \epsilon$.
For larger initial conditions, the lower panels of Figure \ref{fig-2} show that the $A$ and $B$ modes mix up and
that the central pendulum excites large oscillations of other pendula in the chain. As a result, the $\ell^2$ norm
of the solution grows whereas the $\ell^{\infty}$ norm stay at the same level as the initial data.

For $\Omega = 0 \in (-2 \gamma, 2 \gamma + 4 \epsilon)$,
the result of Lemma \ref{lemma-equilibrium} implies the linear instability of the zero equilibrium.
Since the energy methods are not useful to control global dynamics of large solutions
far from the zero equilibrium, the numerical simulations give us the way to explore this phenomenon.
Numerical simulation on Figure \ref{fig-3} show that the growth of oscillation amplitudes
stabilizes at a certain level of the $\ell^{\infty}$ norm and that the $\ell^2$ norm of the solution is significantly larger
than the $\ell^{\infty}$ norm. This indicates that many pendula are excited and oscillate with large and
comparable amplitudes with respect to the central pendulum. Since there are no visible differences between the upper and lower panels,
the transition of the system does not depend on how large the initial amplitude
of the central pendulum is.

\begin{figure}[htb]
\begin{center}
\includegraphics[height=5cm]{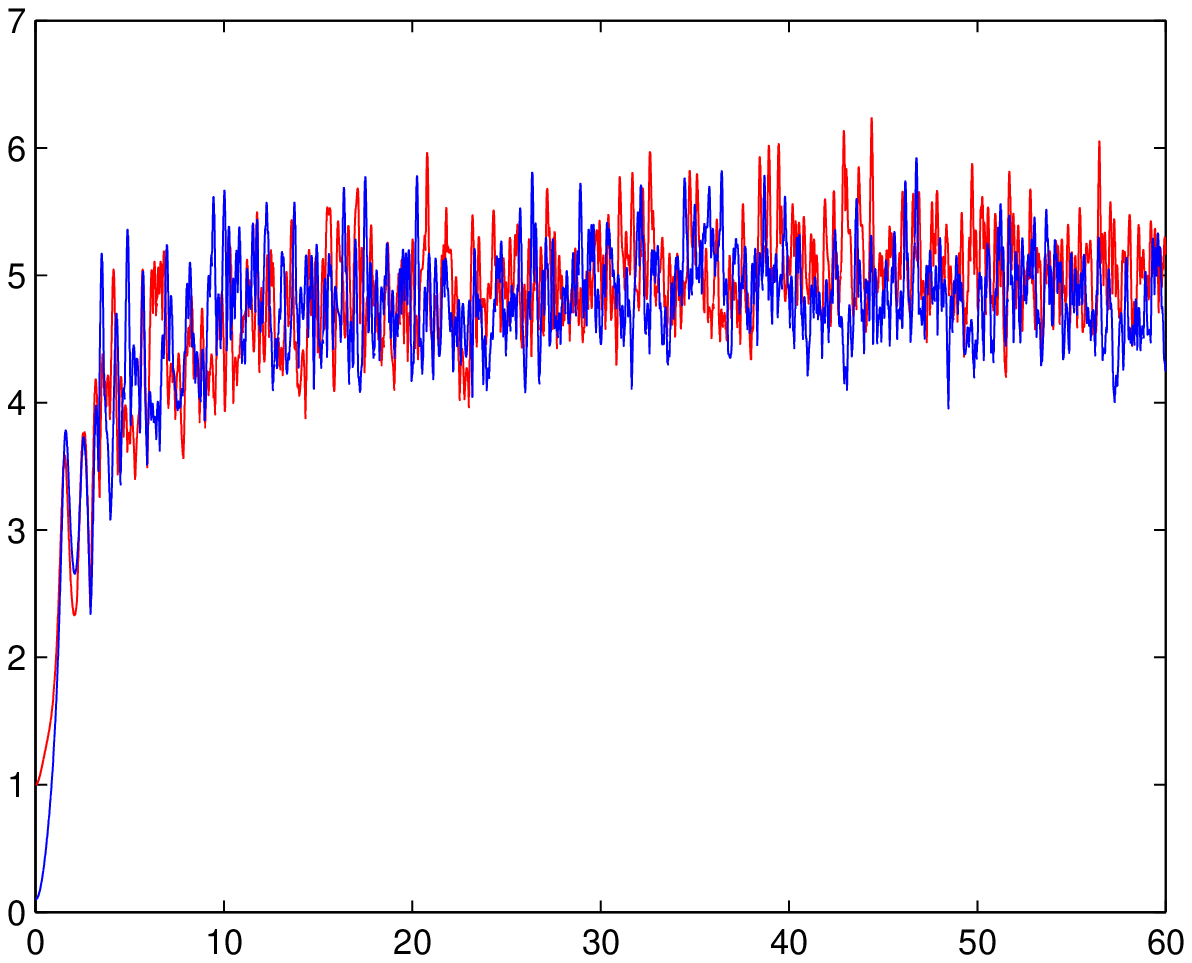}\includegraphics[height=5cm]{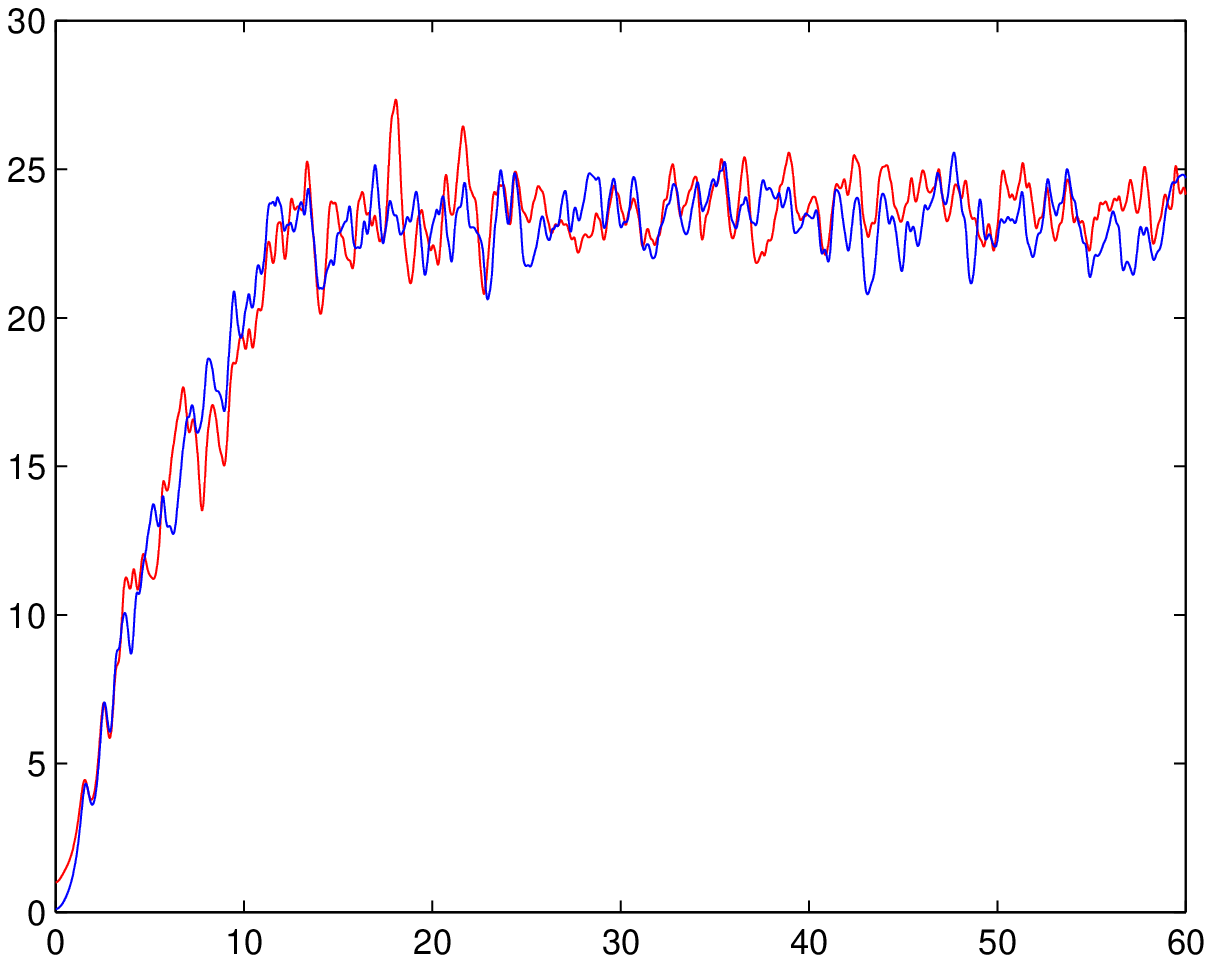}\\
\includegraphics[height=5cm]{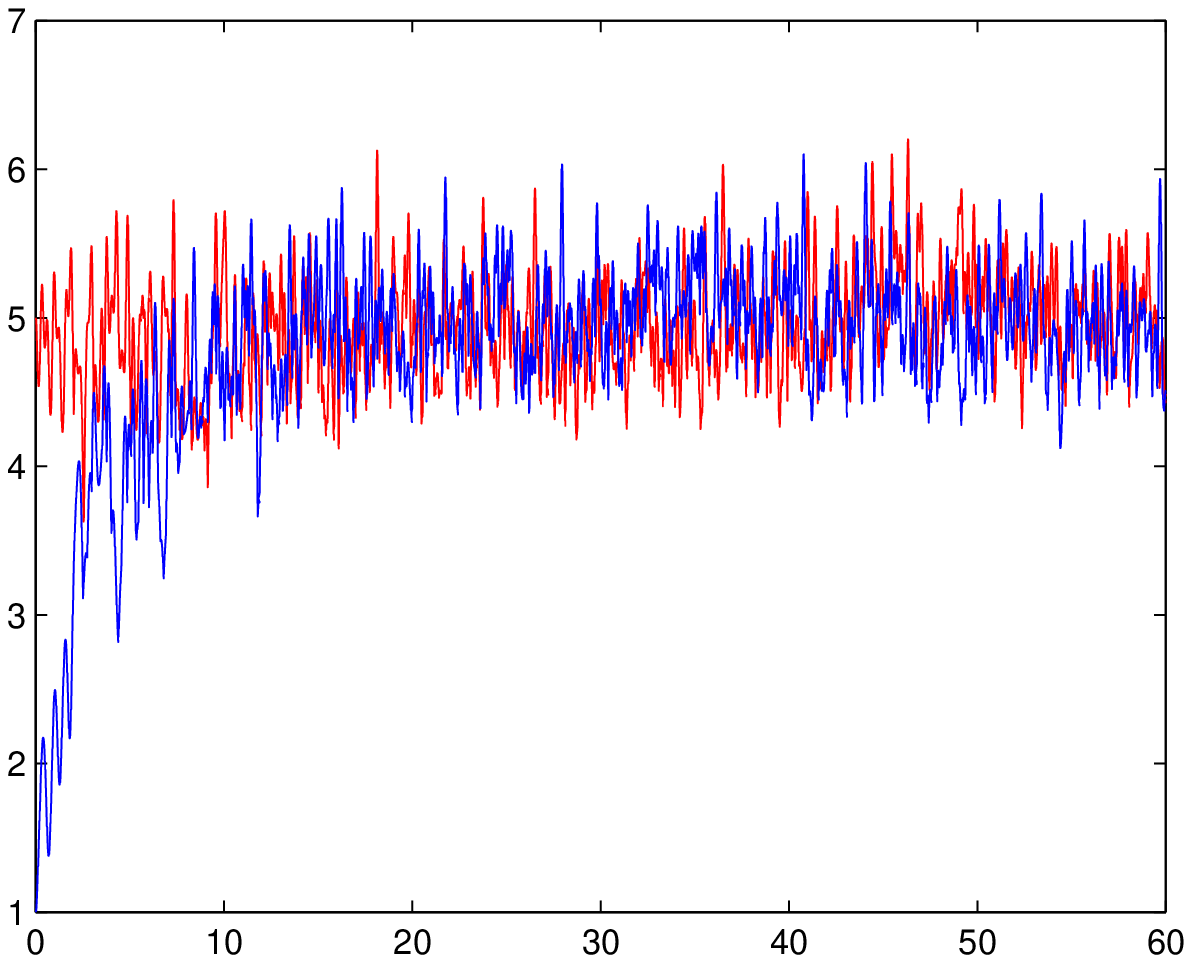}\includegraphics[height=5cm]{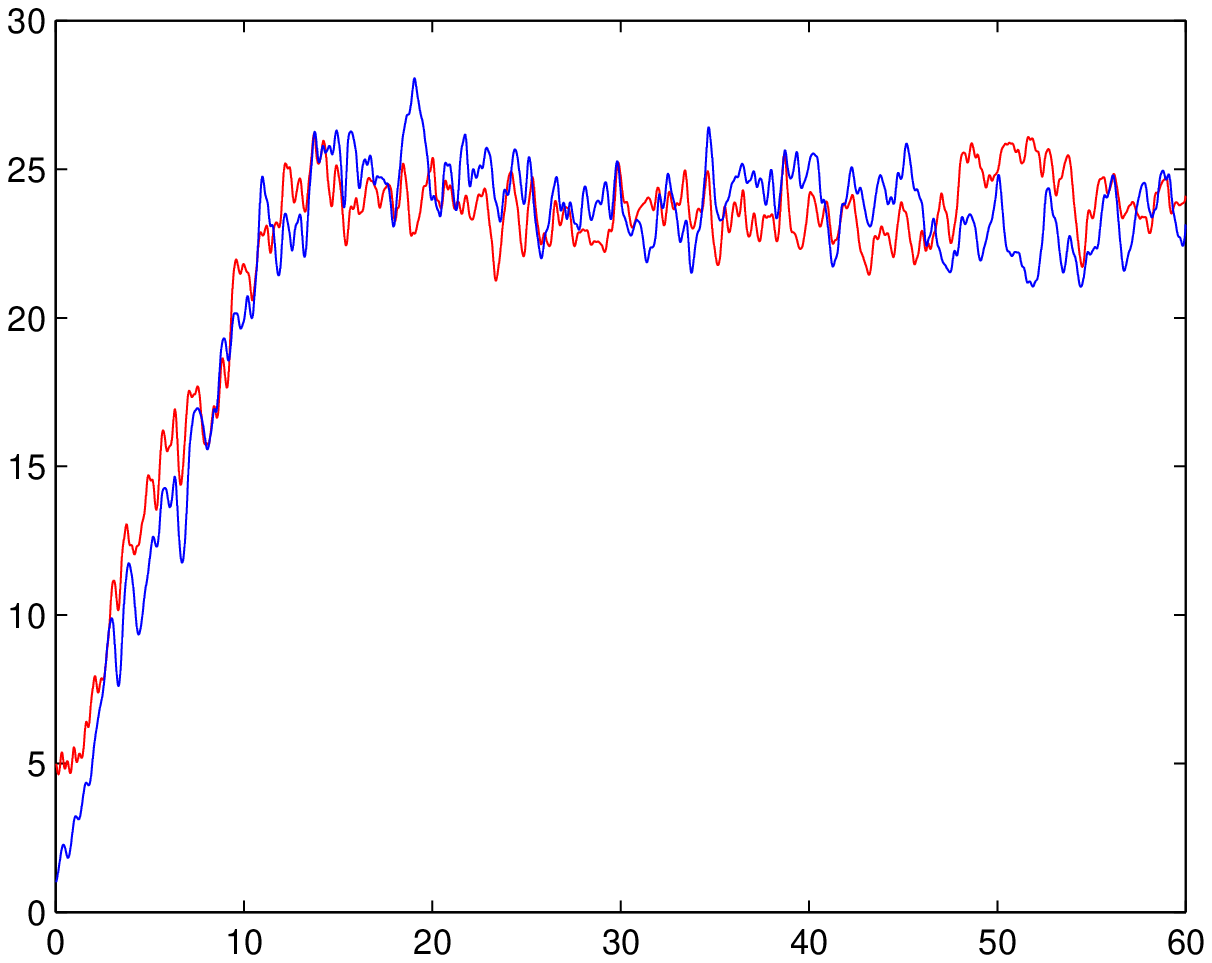}
\end{center}
\caption{$\ell^{\infty}$ (left) and $\ell^2$ (right) norms of the solution for $A$ (red) and $B$ (blue) for $\Omega = 0$:
upper panels for $(A_0,B_0) = (1,0.1)$ and lower panels for $(A_0,B_0) = (5,1)$.}
\label{fig-3}
\end{figure}

By Theorem \ref{theorem-finite}, the solution in the finite chain of $N$ oscillators cannot grow without bounds as $t \to \infty$.
At the same time, the growth of the $\ell^2$ norm of the solution as $t \to \infty$ is not ruled out
for the unbounded lattice in the case $\Omega < 2 \gamma + 4 \epsilon$. In order to illustrate that this
growth does occur, we run numerical experiments, in which we compute the maximal $\ell^2$ norm of the solution
on the time span $[10,50]$ versus $M$, where $N = 2M + 1$ is the number of sites.
Figure \ref{fig-4} (right) shows the nearly monotonic growth of the $\ell^2$ norm (red dots)
of the solution in $M$ in the case $\Omega = 0 \in (-2 \gamma, 2 \gamma + 4 \epsilon)$,
confirming our conjecture on the growth of the solution in the unbounded lattice as $t \to \infty$.
The $\ell^{\infty}$ norm (blue dots) of the solution saturates to the same level independently of $M$.

Figure \ref{fig-4} (left) shows a similar computation for $\Omega = -4 < -2\gamma$ starting with
the smaller initial condition $(A_0,B_0) = (1,0.1)$. In this case, the $\ell^2$ norm of the solution
decreases in $M$ and approaches to a certain limiting level. This indicates that
the growth of the solution on the unbounded lattice does not occur for $\Omega < -2 \gamma$ in the
case of smaller initial amplitude, in agreement with Theorem \ref{theorem-bound-negative}.
Similar dependence (not shown) exists for $\Omega = 8 > 2 \gamma + 4 \epsilon$, when the energy method
excludes growth of the large-norm solutions (see Theorem \ref{theorem-bound}).

\begin{figure}[htb]
\begin{center}
\includegraphics[height=4.5cm]{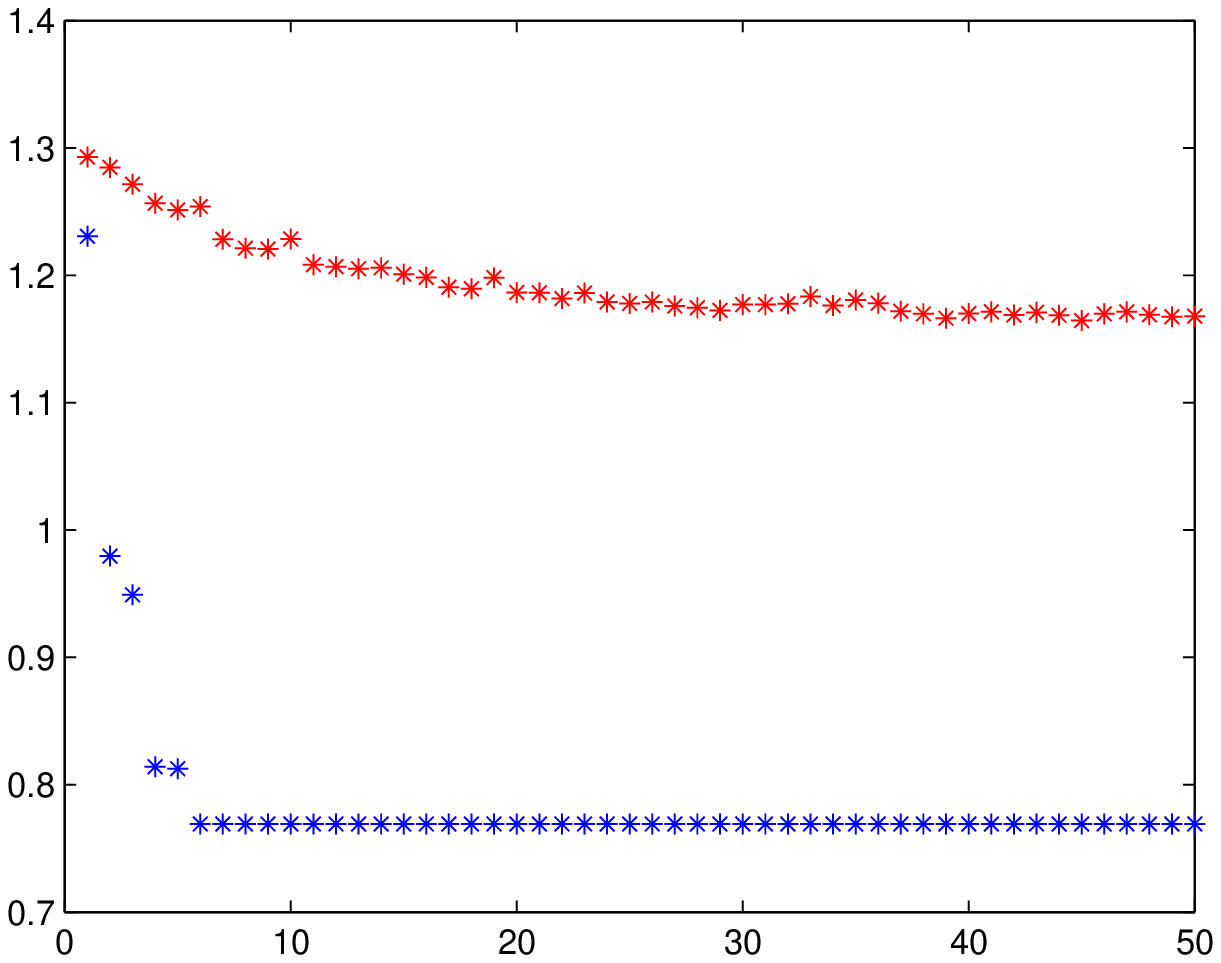}
\includegraphics[height=4.5cm]{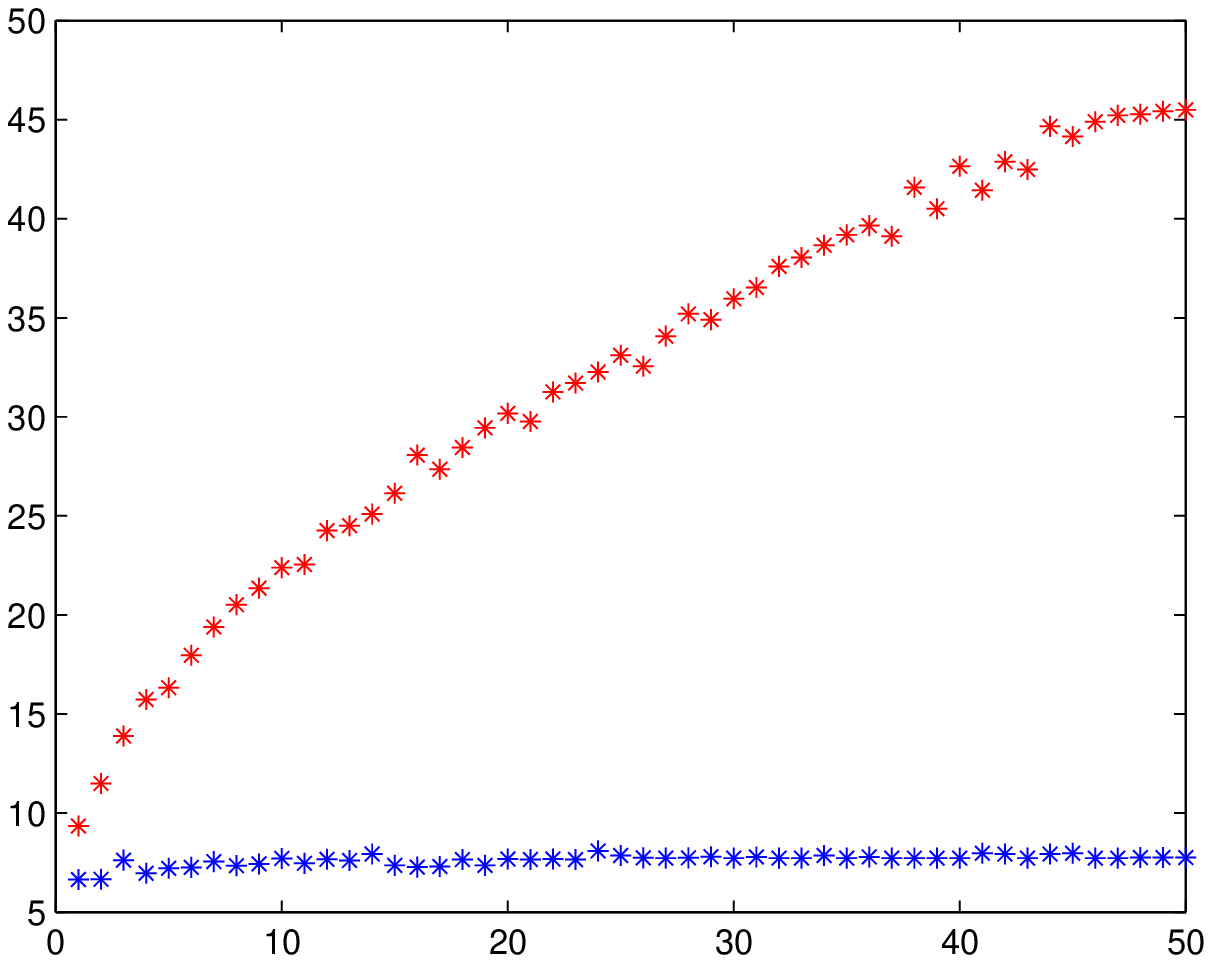}
\end{center}
\caption{Dependence of the $\ell^2$ norm (red dots) and the $\ell^{\infty}$ norm (blue dots)
of the solution averaged over the time interval $[10,50]$ versus
$M$, where $N = 2M+1$ is the number of oscillators, for $\Omega = -4$ (left) and $\Omega = 0$ (right),
in the case $(A_0,B_0) = (1,0.1)$.}
\label{fig-4}
\end{figure}

Similar simulations were carried for initial data with larger amplitudes $(A_0,B_0) = (5,1)$
and they are shown on Figure \ref{fig-5} for $\Omega = 8$ (left) and $\Omega = -4$ (right).
The maximal $\ell^2$ norm is computed on a longer time span $[10,70]$
to provide enough time for the oscillations to be excited in the chain from the central pendulum.
The maximal $\ell^2$ norm does not change versus the number of oscillators $N$
in the case $\Omega = 8 > 2 \gamma + 4 \epsilon$, in agreement with Theorem \ref{theorem-bound}.
At the same time, the maximal $\ell^2$ norm grows with $M$ (and $N$)   for $\Omega = -4 < -2 \gamma$
and apparently diverges as $N \to \infty$. This experiment implies that the limitation of
Theorem \ref{theorem-bound-negative} for small initial data
is not a technical shortfall, the solution $(A,B)$ can grow as $t \to \infty$ in  the unbounded lattice
in the case $\Omega < -2\gamma$.

\begin{figure}[htb]
\begin{center}
\includegraphics[height=4.5cm]{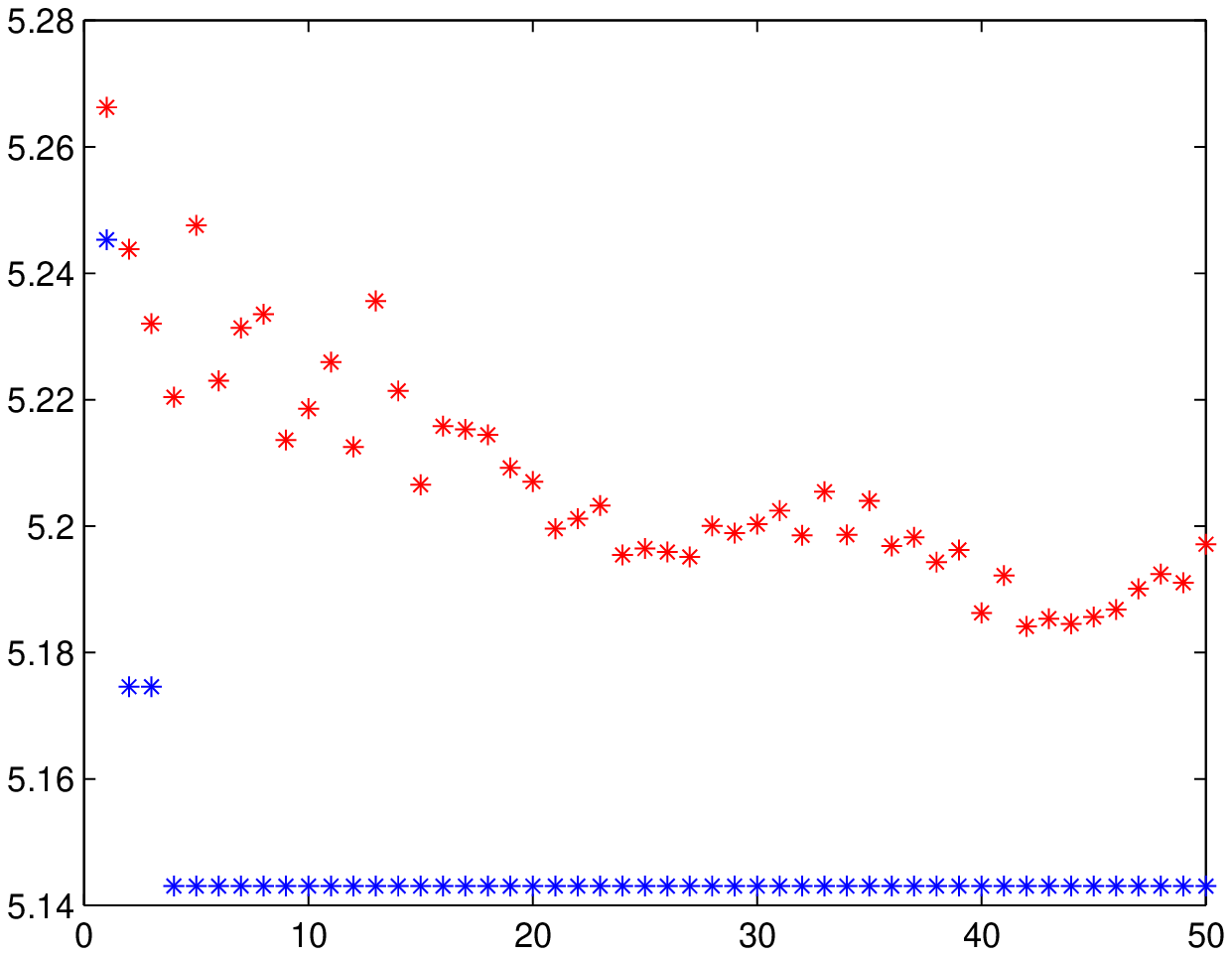}
\includegraphics[height=4.5cm]{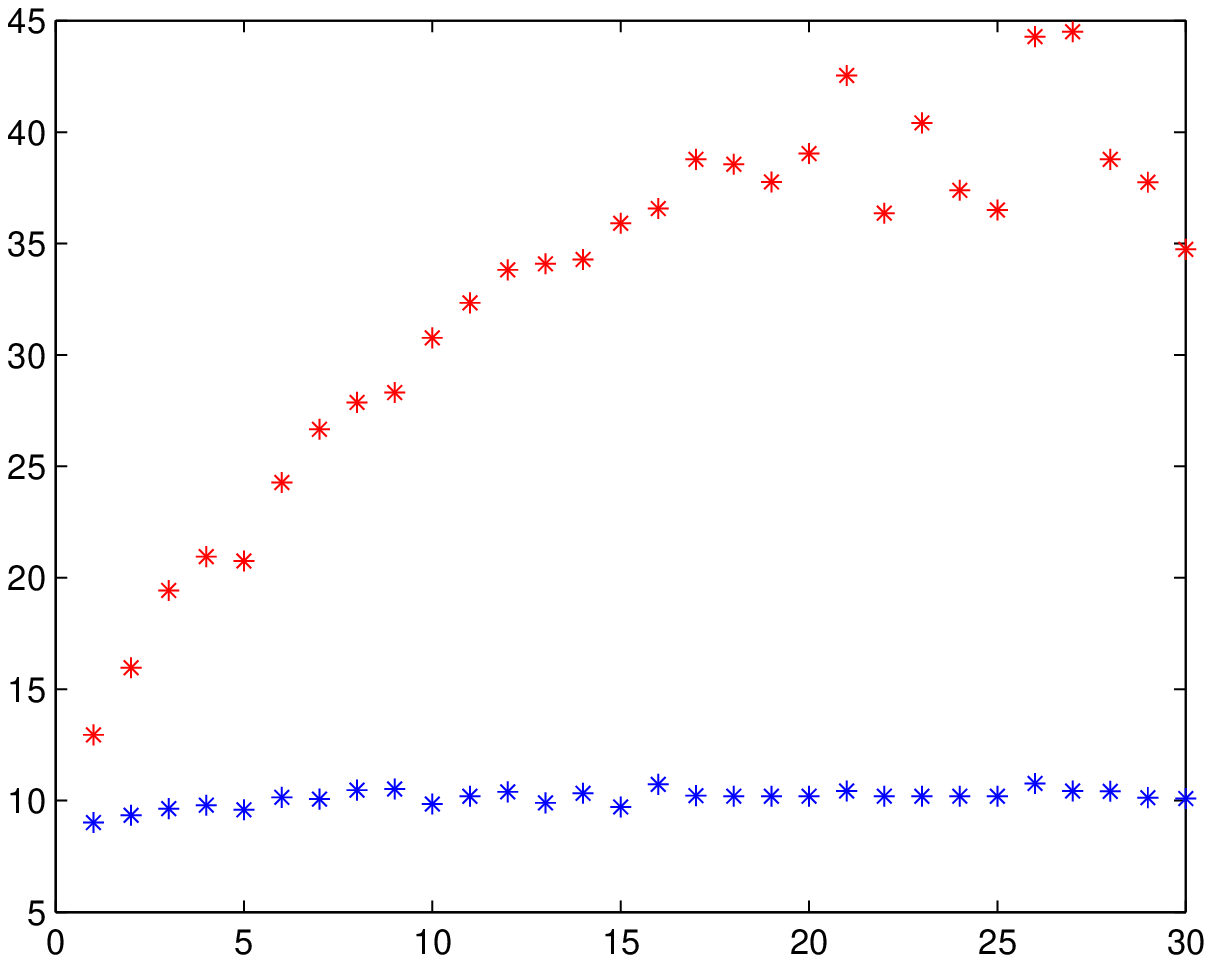}
\end{center}
\caption{Dependence of the $\ell^2$ norm of the solution averaged over the time interval $[10,40]$ versus
$M$, where $N = 2M+1$ is the number of oscillators, for $\Omega = 8$ (left) and $\Omega = -4$ (right),
in the case $(A_0,B_0) = (5,1)$.}
\label{fig-5}
\end{figure}

We note that another numerical solver was also used to approximate dynamics of the dNLS system (\ref{dNLS}).
The alternative solver combines a Crank-Nicholson scheme with the fixed-point iteration method
to solve the coupled system of differential equations and it was implemented in
Python and Fortran. The numerical experiments shown on Figures \ref{fig-1}--\ref{fig-5}
were reproduced on this other solver without any significant differences.

\section{Conclusion}

We have addressed a novel model of the $\mathcal{PT}$-symmetric dNLS equation, which
admits Hamiltonian formulation but exhibits no gauge symmetry. The model describes
dynamics in the chain of coupled pendula pairs with the small diffusive couplings
near the $1:2$ resonance between the parameterically driven force and the linear
frequency of each pendulum.

In the context of the problem of nonlinear stability of zero equilibrium,
we have obtained sharp bounds on the $\ell^2$ norms of the oscillation amplitudes
in terms of parameters of the model and the size of the initial condition.
Stable dynamics of the coupled pendula chains is relatively simple, but it
becomes more interesting when the stability constraints are not satisfied.
In the latter cases, we show that the central pendulum excites nearest pendula such
that the $\ell^2$ norm of the oscillation amplitudes grows while the $\ell^{\infty}$ norm
remains finite.

As the next stage in this work, it may be interesting to characterize the breather
(periodic or quasi-periodic) solutions of the model and to see how existence and stability
of such solutions is related to the nonlinear stability of the zero equilibrium.
Another open question is to address stable and unstable dynamics in the original
Newton's equations of motion and to compare it with predictions of the reduced dNLS model.

\vspace{0.25cm}

\noindent{\bf Acknowledgements.}
This work was developed during the visit of D.E. Pelinovsky to Universit\'{e} des Antilles. D.P. thanks members
of the institute for hospitality during this visit. The authors thank A. Chernyavsky for preparing Fig. \ref{fig-1}.


\begin{thebibliography}{99}


\bibitem{barashenkov} I.V. Barashenkov and M. Gianfreda,
{\em An exactly solvable PT-symmetric dimer from a Hamiltonian system
of nonlinear oscillators with gain and loss},
J. Phys. A: Math. Theor. {\bf 47}, 282001 (2014).

\bibitem{barashenkovInt} I.V. Barashenkov, D.E. Pelinovsky and P. Dubard,
{\em Dimer with gain and loss: Integrability and-symmetry restoration},
J. Phys. A: Math. Theor. {\bf 48}, 325201 (2015).

\bibitem{bender} C.M. Bender,
{\em Making sense of non-Hermitian Hamiltonians},
Rep. Prog. Phys. \textbf{70}, 947 (2007).

\bibitem{benderExp} C.M. Bender, B. Berntson, D. Parker and E. Samuel,
{\em Observation of PT phase transition in a simple mechanical system},
Am. J. Phys. {\bf 81}, 173 (2013).

\bibitem{Kivshar} O.M Braun and Y. Kivshar,
{\em The Frenkel-Kontorova model: concepts, methods, and applications},
(Springer-Verlag Berlin Heidelberg, 2004).

\bibitem{ChernPel1} A. Chernyavsky and D.E. Pelinovsky,
{\em Breathers in Hamiltonian $\mathcal{PT}$-symmetric chains of coupled pendula under a resonant periodic force},
Symmetry {\bf 8} (2016), 59 (26 pages).

\bibitem{ChernPel2} A. Chernyavsky and D.E. Pelinovsky,
{\em Long-time stability of breathers in Hamiltonian $\mathcal{PT}$-symmetric lattices},
J. Phys. A {\bf 49} (2016), 475201 (20 pages).

\bibitem{Xu1} J Cuevas, L.Q. English, P.G. Kevrekidis, and M. Anderson,
Phys. Rev. Lett. {\bf 102} (2009), 224101 (5 pages).

\bibitem{Kevrekidis} P.G. Kevrekidis, {\em The Discrete Nonlinear Schr\"{o}dinger Equation}
(Springer-Verlag, Berlin, Heidelberg, 2009).

\bibitem{pel1} P.G. Kevrekidis, D.E. Pelinovsky and D.Y. Tyugin,
{\em Nonlinear dynamics in PT-symmetric lattices}, J. Phys. A: Math. Theor. {\bf 46}, 365201 (2013).

\bibitem{Xu2} F. Palmero, J.Han, L.Q. English, T.J. Alexander, and P.G. Kevrekidis,
{\em Multifrequency and edge breathers in the discrete sine--Gordon system via subharmonic driving: Theory, computation and experiment},
Phys. Lett. A {\bf 380} (2016), 402--407.

\bibitem{PPP} D. Pelinovsky, T. Penati and S. Paleari,
{\em Approximation of small-amplitude weakly coupled oscillators with discrete nonlinear
Schr\"{o}dinger equations}, Rev. Math. Phys. {\bf 28} (2016), 1650015 (25 pages).

\bibitem{pel3}   D.E. Pelinovsky, D.A. Zezyulin and V.V. Konotop,
{\em Nonlinear modes in a generalized $\mathcal{PT}$-symmetric discrete nonlinear Schrodinger equation},
J. Phys. A: Math. Theor. {\bf 47}, 085204 (2014).

\bibitem{pel4}   D.E. Pelinovsky, D.A. Zezyulin and V.V. Konotop,
{\em Global existence of solutions to coupled $\mathcal{PT}$-symmetric nonlinear Schrodiger equations},
Int. J. Theor. Phys. {\bf 54} (2015), 3920--3931.

\bibitem{Pikovsky} A. Pikovsky, M. Rosenblum and J. Kurths,
{\em Synchronization: a universal concept in nonlinear sciences},
Cambridge Nonlinear Science Series {\bf 12} (Cambridge University Press, Cambridge, 2003).

\bibitem{pascal1} E. Destyl, S.P. Nuiro, and P. Poullet,
{\em On the global behavior of solutions of a coupled system of nonlinear Schr\"{o}dinger equation},
Stud. Appl. Math. {\bf 138}  (2017),  227--244.

\bibitem{schindlerExp} J. Schindler, A. Li, M.C. Zheng, F.M. Ellis, T. Kottos,
{\em Experimental study of active LRC circuits with PT symmetries},
Phys. Rev. A {\bf 84}, 040101 (2011).

\bibitem{Susanto1} H. Susanto, Q.E. Hoq and P.G. Kevrekidis, {\em Stability of discrete solitons
in the presence of parameteric driving}, Phys. Rev. E {\bf 74}, 067601 (2006).

\bibitem{Susanto2} M. Syafwan, H. Susanto and S. M. Cox, {\em Discrete solitons in electromechanical resonators},
Phys. Rev. E {\bf 81}, 026207 (2010).

\bibitem{Xu3} Y. Xu, T.J. Alexander, H. Sidhu, and P.G. Kevrekidis,
{\em Instability dynamics and breather formation in a horizontally shaken pendulum chain},
Phys. Rev. E {\bf 90} (2014), 042921 (11 pages).
\end{thebibliography}
\end{document}